\documentclass[a4paper,conference]{IEEEtran}

\usepackage[utf8]{inputenc}
\usepackage[T1]{fontenc}
\usepackage{url}
\usepackage{ifthen}
\usepackage[cmex10]{amsmath} 

\usepackage{amsmath}
\usepackage{amssymb}
\usepackage{amsthm}
\usepackage{mathrsfs}
\usepackage{graphicx}
\usepackage{color}
\usepackage[colorlinks=true,citecolor=blue,linkcolor=blue]{hyperref}
\usepackage[linesnumbered,ruled,vlined,titlenumbered]{algorithm2e}

\newcommand{\FF}{\mathbb{F}}
\newcommand{\Fq}{\mathbb{F}_q}
\newcommand{\Fqm}{\mathbb{F}_{q^{m}}}
\newcommand{\C}{\mathbb{C}}

\newtheorem{theorem}{Theorem}
\newtheorem{lemma}[theorem]{Lemma}

\newtheorem{definition}[theorem]{Definition}

\newtheorem{remark}[theorem]{Remark} 

\usepackage{varioref}        
\usepackage{fancyref}

\makeatletter
\def\mkfancyprefix#1#2{%
\expandafter\def\csname fancyref#1labelprefix\endcsname{#1}%
\begingroup\def\x{\endgroup\frefformat{plain}}%
    \expandafter\x\csname fancyref#1labelprefix\endcsname
    {\MakeLowercase{#2}\fancyrefdefaultspacing##1}%
\begingroup\def\x{\endgroup\Frefformat{plain}}%
    \expandafter\x\csname fancyref#1labelprefix\endcsname
    {#2\fancyrefdefaultspacing##1}%
\begingroup\def\x{\endgroup\frefformat{vario}}%
    \expandafter\x\csname fancyref#1labelprefix\endcsname
    {\MakeLowercase{#2}\fancyrefdefaultspacing##1##3}%
\begingroup\def\x{\endgroup\Frefformat{vario}}%
    \expandafter\x\csname fancyref#1labelprefix\endcsname
    {#2\fancyrefdefaultspacing##1##3}%
}
\makeatother
\fancyrefchangeprefix{\fancyrefeqlabelprefix}{eqn}
\mkfancyprefix{ssec}{Section}
\mkfancyprefix{tbl}{Table}
\mkfancyprefix{thm}{Theorem}
\mkfancyprefix{lem}{Lemma}
\mkfancyprefix{cor}{Corollary}
\mkfancyprefix{prop}{Proposition}
\mkfancyprefix{prob}{Problem}
\mkfancyprefix{alg}{Algorithm}
\mkfancyprefix{inv}{Invariant}
\mkfancyprefix{ex}{Example}
\mkfancyprefix{line}{Line}
\mkfancyprefix{def}{Definition}
\mkfancyprefix{itm}{Item}
\mkfancyprefix{app}{Appendix}
\mkfancyprefix{rem}{Remark}
\newcommand{\cref}[1]{\Fref{#1}}

\makeatletter
\newcommand{\removelatexerror}{\let\@latex@error\@gobble}
\makeatother

\newcommand{\printalgoIEEE}[1]
{{\centering
\scalebox{0.97}{
\removelatexerror
\begin{tabular}{p{\columnwidth}}
\begin{algorithm}[H]
 #1
\end{algorithm}
\end{tabular}
}
}
}

\renewcommand{\vec}[1]{\ensuremath{\mathbf{#1}}}

\newcommand{\N}{\mathbb{N}}
\newcommand{\U}{\mathcal{U}}
\newcommand{\A}{\mathcal{A}}
\newcommand{\CGab}{\mathcal{C}_\mathrm{G}}

\newcommand{\autom}{\theta}
\newcommand{\K}{K}
\renewcommand{\L}{L}
\newcommand{\Lset}{\L[x;\theta]}
\newcommand{\Gal}[1]{\mathrm{Gal}\left(#1\right)}
\newcommand{\GalLK}{\Gal{\L/\K}}

\newcommand{\ev}{\mathrm{ev}}
\newcommand{\charpoly}{\mathrm{char}}
\newcommand{\id}{\mathrm{id}}

\newcommand{\Bb}{\mathcal{B}}

\newcommand{\MSP}[1]{\mathcal{A}_{#1}}
\newcommand{\BigO}[1]{O(#1)}

\newcommand{\rhat}{\hat{r}}
\newcommand{\terr}{\tau}
\newcommand{\wtR}{\mathrm{wt}_\mathrm{R}}

\renewcommand{\r}{\vec{r}}
\renewcommand{\c}{\vec{c}}
\newcommand{\e}{\vec{e}}
\newcommand{\LH}[1]{\langle #1 \rangle}

\newcommand{\quo}{\chi}
\newcommand{\rem}{\varrho}

\newcommand{\rk}{\mathrm{rank}}

\newcommand{\smallsum}{{\textstyle\sum\nolimits}}

\newcommand{\Gpoly}{\MSP{\LH{g_1,\dots,g_n}}}

\begin{document}
\title{An Alternative Decoding Method for Gabidulin Codes in Characteristic Zero}

\author{\IEEEauthorblockN{Sven Müelich\IEEEauthorrefmark{1},
Sven Puchinger\IEEEauthorrefmark{1},
David Mödinger\IEEEauthorrefmark{2},
Martin Bossert\IEEEauthorrefmark{1}}
\IEEEauthorblockA{\IEEEauthorrefmark{1}
              Ulm University, 
              Institute of Communications Engineering, 
              89081 Ulm, Germany\\
              Email: \{sven.mueelich, sven.puchinger, martin.bossert\}@uni-ulm.de}
\IEEEauthorblockA{\IEEEauthorrefmark{2}
              Ulm University, 
              Institute of Distributed Systems, 
              89081 Ulm, Germany\\
              Email: david.moedinger@uni-ulm.de}
}

\maketitle

\begin{abstract}
Gabidulin codes, originally defined over finite fields, are an important class of rank metric codes with various applications.
Recently, their definition was generalized to certain fields of characteristic zero and a Welch--Berlekamp like algorithm with complexity $O(n^3)$ was given.
We propose a new application of Gabidulin codes over infinite fields: low-rank matrix recovery.
Also, an alternative decoding approach is presented based on a Gao type key equation, reducing the complexity to at least $O(n^2)$.
This method immediately connects the decoding problem to well-studied problems, which have been investigated in terms of coefficient growth and numerical stability. {\let\thefootnote\relax\footnote{{This work has been supported by DFG, Germany, under grant BO 867/32-1.}}}
\end{abstract}


\section{Motivation}
Finding a matrix of minimal rank is a problem which occurs in different scenarios. For example in random linear network coding~\cite{koetter2008}, an error can be described by a matrix of minimal rank. 
Therefore, codes whose metric is based on the rank of matrices can be beneficial. 
The most prominent example of rank metric codes are Gabidulin codes, introduced by Delsarte \cite{delsarte1978bilinear}, Gabidulin \cite{gabidulin1985theory}, and Roth \cite{roth1991maximum}.
Given a received word $\mathbf{R} = \mathbf{C} + \mathbf{E}$, the calculation of the error matrix $\mathbf{E}$ of minimum rank can be described by the weight-minimization problem
\begin{align}
 \min  \rk(\mathbf{E'}) ~\text{subject to } \mathbf{HE'} = \mathbf{HE},
 \label{EqMinimization}
\end{align}
where $\mathbf{H}$ is a parity check matrix.
This minimization problem is equivalent to the problem of low-rank matrix recovery (LRMR) \cite{candes2009,gross2011recovering}, which is the matrix-analogue to compressed sensing \cite{candes2006,donoho2006}. 
This problem aims to recover an unknown matrix $\mathbf{E} \in \C^{n \times n}$ of low rank, and can be solved by finding a solution for the under-determined linear system of equations $\mathbf{He} = \mathbf{s}$, where $\mathbf{H}\in \C^{m \times n^2}$ is the sensing matrix, $\mathbf{e} \in \C^{{n^2}\times 1}$ is the vector representation of the matrix $\mathbf{E}$, and $\mathbf{s}\in \C^{m \times 1}$ is the measurement when applying the sensing matrix $\mathbf{H}$ to $\mathbf{E}$ ($m<n^2$).
Applications of LRMR can be found e.g., in the fields of image processing or collaborative filtering. 
Since decoding of rank metric codes and LRMR is the same mathematical problem (cf. Equation~(\ref{EqMinimization})), the application of Gabidulin codes in characteristic zero might be promising to the LRMR problem.
If we replace the rank metric by the Hamming metric, Equation~(\ref{EqMinimization}) describes both a Hamming-metric decoder and the compressed sensing problem.
An exchange of concepts between these two areas was successfully investigated in the recent years \cite{zorlein2015}.
Another important application of Gabidulin codes in characteristic zero is space-time coding.

Commonly, Gabidulin codes are defined over finite fields as evaluation codes of linearized polynomials and can be considered as rank metric equivalents of Reed-Solomon codes. 
In \cite{mohamed2015deterministic}, Reed-Solomon codes over the complex field were investigated for applications in compressed sensing.
LRMR and space-time codes indicate that there is a need for Gabidulin codes defined over fields of characteristic zero, possibly dense in $\mathbb{C}$.
In \cite{augot2013rank} and \cite{robert2015phd}, Gabidulin codes in characteristic zero were introduced.
In contrast to the finite field case, $\autom$-polynomials are used instead of linearized polynomials. 
A Welch-Berlekamp-like decoding algorithm \cite{loidreau2006} was transformed from the finite field case to the characteristic zero case, which allows decoding in cubic time. 
In this work, we consider an alternative method for decoding characteristic zero Gabidulin codes.

The rest of the paper is structured as follows:
Section~\ref{sec:Gabidulin} outlines Gabidulin codes and related concepts in characteristic zero.
In Section~\ref{sec:Decoding} we propose a new decoding approach.
We explain how the decoding problem can be solved by using shift register synthesis to find solutions of a Gao-like key equation.
We also discuss issues of coefficient growth and numerical problems which emerge when using infinite fields.
Finally, Section~\ref{sec:Conclusion} concludes the paper.

\section{Gabidulin Codes Over Infinite Fields}
\label{sec:Gabidulin}

This section first summarizes properties of $\autom$-polynomials, which are used to define Gabidulin codes in characteristic zero. 
Then we recall different definitions of rank metric and the definition of Gabidulin codes.

\subsection{$\autom$-polynomials}

Gabidulin codes over finite fields are usually defined using \emph{linearized polynomials} \cite{ore1933special}.
$\autom$-polynomials can be seen as a natural generalization of linearized polynomials for arbitrary fields.
Let $K \subseteq L$ be fields and $L/K$ be a Galois extension.
The \emph{Galois group} of $L/K$ is given by
\begin{align*}
\GalLK = \left\{ \autom : \L \to \L \;\text{automorphism} : \autom(k) = k \; \forall k \in \K \right\}.
\end{align*}

\begin{lemma}{\cite{ore1933theory}}\label{lem:theta_poly_ring}
Let $\autom \in \GalLK$. The set
\begin{align*}
\Lset = \left\{a = \smallsum_{i=0}^{d_a} a_i x^i : a_i \in \L, \, d_a \in \N, \, a_{d_a} \neq 0 \right\}
\end{align*}
with multiplication rule $x \cdot \alpha = \autom(\alpha) \cdot x$ for all $\alpha \in \L$, extended to polynomials inductively, and ordinary addition is a ring.
\end{lemma}
We call the polynomial ring of \cref{lem:theta_poly_ring} a $\autom$-polynomial ring.
The degree of $a \in \Lset$ is given by $\deg a = d_a$ and
$a$ is called \emph{monic} if $a_{d_a} = 1$.

\begin{remark}
We state the following properties of $\Lset$.
\begin{itemize}
\item $(\Lset,+,\cdot)$ is non-commutative in general.
\item $\autom$-polynomials are a special case of skew polynomials~\cite{ore1933theory} with derivation $\delta=0$.
\item For $\K=\Fq$, $\L=\Fqm$ and the Frobenius automorphism $\autom = \cdot^q$, $\Lset$ is isomorphic to 
a linearized polynomial ring. Note that $\cdot^q \in \Gal{\Fqm/\Fq}$.
\end{itemize}
\end{remark}

Is was already proven in \cite{ore1933special} that $\Lset$ is a left- and right- Euclidean domain. E.g., the following division lemma is true.
\begin{lemma}\label{lem:division}\cite{ore1933special}
For $a\in \Lset$, $b \in \Lset^{*}$, $\exists$~unique $\quo,\rem\in\Lset$: $a = \quo \cdot b + \rem$ (right division), where $\deg \rem  < \deg b$.
\end{lemma}
Related to division, we can define the (right) modulo congruence relation for $a,b,c \in \Lset$:
\begin{align*}
a \equiv b \mod c \quad :\Leftrightarrow \quad \exists d \in \Lset : a = b + d \cdot c.
\end{align*}

We can define an evaluation map\footnote{There are several definitions of evaluation maps for $\theta$-polynomials, cf.~\cite{boucher2014linear} for the general skew polynomial case.} on $\Lset$ as
\begin{align}
\label{eqn:ev}
\ev_a = a(\cdot) : \L \to \L, \quad \alpha \mapsto \smallsum_{i=0}^{d_a} a_i \autom^i(\alpha),
\end{align}
where $\autom^i(\cdot) = \underset{i \text{ times}}{\underbrace{\autom(\autom(\dots\autom(}}\cdot)\dots))$.
From $\theta \in \GalLK$ it follows that $\autom : L \to L$ is a $\K$-linear map.
Thus, also $\ev_a$ is a linear map and the root space of a $\autom$-polynomial $a$,
\begin{align*}
\ker(a) = \{\alpha \in \L : a(\alpha)=0 \},
\end{align*}
is a linear subspace of $\L$.
The evaluation map of the multiplication of two $\autom$-polynomials $a,b$ equals the composition of the evaluation maps of $a,b$ respectively, i.e. $\ev_{a \cdot b} = \ev_a \circ \ev_b$.
Since $\autom$ is a linear map, it has well-defined eigenvalues which are the roots of its characteristic polynomial
\begin{align*}
\charpoly_\autom(x) = \det(x \cdot \id_\L - \autom).
\end{align*}
The eigenvalues and characteristic polynomial are the same as of any matrix representation of $\autom$ in a basis of $\L$ over $\K$.
We say that a characteristic polynomial is square-free if all its roots have multiplicity one.
If $\charpoly_\autom$ is square-free, $\autom$ has distinct eigenvalues and any of its matrix representations is diagonalizable.
Using these properties, we can state the following theorem.

\begin{lemma}{\cite[Theorem 6]{robert2015phd}}\label{lem:dimker_deg_bound}
If $\charpoly_\autom$ is square-free, then
\begin{align*}
\dim_\K(\ker(a)) \leq \deg(a) \quad \forall a \in \Lset \setminus \{0\}
\end{align*}
\end{lemma}

\begin{proof}
The proof can be found in {\cite[Theorem~6]{robert2015phd}}.
It uses matrix representations of $\autom$ and the fact that it is diagonalizable due to $\charpoly_\autom$ being square-free.
\end{proof}

\begin{theorem}\label{thm:existence_annihilator}
Let $\U \subseteq \L$ be an $s$-dimensional $\K$-subspace. If $\charpoly_\autom$ is square-free, there exists a unique monic $\autom$-polynomial $\MSP{\U}$ with $\U \subseteq \ker(\MSP{\U})$ of minimum degree.
$\MSP{\U}$ is called \emph{annihilator polynomial} of $\U$ and if $\autom(\cdot)$ can be calculated in $O(1)$, $\MSP{\U}$ can be computed in $\BigO{s^2}$ operations in $\L$.
Moreover, $\deg \MSP{\U}  = \dim \U$ and $\U = \ker(\MSP{\U})$.
\end{theorem}

\begin{proof}
The proof is similar to \cite[Theorem~8]{robert2015phd}.
It can be shown by induction that the polynomial $\A_s$ constructed in \cref{alg:annihilator} fulfills $\U \subseteq \ker(\A_s)$.
By the Euclidean algorithm, $\A_s = \quo \cdot \MSP{\U}$ for some $\quo \in \Lset$ and $\deg \MSP{\U} \leq \deg \A_s = \dim \U$ because $\deg A_i = \deg A_{i-1}+1$ $\forall i$, $\deg A_0 = 0$ and thus $\deg \A_s = s = \dim \U$.
Also, $\deg \MSP{\U} \geq \dim \U$ by \cref{lem:dimker_deg_bound} (which assumes that $\charpoly_\autom$ is square-free), implying $\deg \MSP{\U} = \dim \U$.
Since $\MSP{\U}$ is defined to be monic, it is therefore unique and $\A_s = \MSP{\U}$.
Due $\dim(\ker(\MSP{\U})) \leq \deg \MSP{\U} = \dim \U$, together with $\U \subseteq \ker(\MSP{\U})$, it follows that $\U =\ker(\MSP{\U})$.
Line~\ref{line:ann_2} of \cref{alg:annihilator} is executed $s$ times and each loop requires
\begin{itemize}
\item one evaluation $A_{i-1}(u_i)$, costing $\BigO{s}$ operations in $\L$ by naively applying the evaluation formula~\ref{eqn:ev}
\item one computation of $\theta(A_{i-1}(u_i))$ $\Rightarrow$ $\BigO{1}$ and
\item one addition in $\Lset$ $\Rightarrow$ $\BigO{s}$,
\end{itemize}
and hence the algorithm has complexity $\BigO{s}$ in $\L$.
\end{proof}

\vspace{-1em}
\printalgoIEEE{
\DontPrintSemicolon
\KwIn{$\K$-basis $(u_1,\dots,u_s)$ of $\U \subseteq \L$.}
\KwOut{$\MSP{\U}$ as in \cref{thm:existence_annihilator}.}
$\A_0 \gets 1$ \; \label{line:ann_1}
\For{$i=1,\dots,s$}{
$\A_i \gets (x-\frac{\theta(A_{i-1}(u_i))}{A_{i-1}(u_i)}) \cdot A_{i-1}$ \hfill \tcp{$\BigO{s}$}  \label{line:ann_2}
}
\Return{$\A_s$}  \label{line:ann_4}
\caption{Annihilator Polynomial \cite{robert2015phd}}
\label{alg:annihilator}
}

\begin{theorem}[{\cite[Theorem~8]{robert2015phd}}]\label{thm:interpolation}
Let $g_{1},\dots, g_{n} \in \L$, linearly independent over $\K$, and $\r = (r_{1},\dots,r_{n}) \in \L^{n}$. Then there is a unique monic $\autom$-polynomial $\rhat$ of degree $n-1$ such that
\begin{align*}
\hat{r}(g_{i}) = r_{i} \quad \forall i=1,\dots,n.
\end{align*}
\end{theorem}

\subsection{Rank Metric in Characteristic Zero}

Let $\K \subseteq \L$ be fields, $\L/\K$ a Galois extension of degree $m$ and $\Bb$ a basis of $L$ over $K$. The number of $k$-linearly independent columns of a matrix $\mathbf{X}$ is denoted by $\rk_k(\mathbf{X})$ for $k \in \{L,K\}$. We define the matrices
\arraycolsep=0pt
\begin{align*}
\mathbf{X_{\theta}} = 
\begin{pmatrix}
x_{1} & \dots & x_{n} \\
\theta(x_{1}) & \dots & \theta(x_{n}) \\
\vdots & \ddots & \vdots \\
\theta^{m-1}(x_{1}) & \dots &  \theta^{m-1}(x_{n})
\end{pmatrix}, \,
\mathbf{X_{\Bb}} = 
\begin{pmatrix}
x_{1,1} & \dots & x_{n,1} \\
x_{1,1} & \dots & x_{n,1} \\
\vdots & \ddots & \vdots \\
x_{1,m} & \dots &  x_{n,m}
\end{pmatrix},
\end{align*}
where $(x_{i,1},\dots,x_{i,m})^T \in K^m$ is the representation of $x_i \in L$ in the basis $\Bb$.
In \cite[Section~2.2]{robert2015phd} four definitions of rank weight in characteristic zero are given. 
\begin{definition}[\cite{robert2015phd}]\label{def:rank_metrics}
Let $\mathbf{x} \in \L^n$. We define the rank weights
\begin{align*}
\omega_{1}(\mathbf{x}) &= deg(\MSP{\LH{x_1,\dots,x_n}}) \\
\omega_{2}(\mathbf{x}) &= \rk_\L(\mathbf{X_{\theta}}) \\
\omega_{3}(\mathbf{x}) &= \rk_\K(\mathbf{X_{\theta}}) \\
\omega_{4}(\mathbf{x}) &= \rk_\K(\mathbf{X_{\Bb}})
\end{align*}
The corresponding rank metrics can be defined as 
\begin{align*}
\mathrm{d}_{\mathrm{R},i}(\mathbf{x},\mathbf{y}) = \omega_i (\mathbf{x}-\mathbf{y}) \quad \forall i \in \{1,2,3,4\}.
\end{align*}
\end{definition}

In the finite field case, these rank weights are the same.
Over characteristic zero, the following relation can be proven.

\begin{lemma}{\cite[Lemmata~13, 14, and 15]{robert2015phd}}\label{lem:rank_weights}
\begin{align*}
\omega_{1}(\mathbf{x}) = \omega_{2}(\mathbf{x}) \leq \omega_{3}(\mathbf{x}) = \omega_{4}(\mathbf{x}) 
\end{align*}
\end{lemma}

\subsection{Gabidulin Codes}

Gabidulin codes were originally defined by \cite{gabidulin1985theory,delsarte1978bilinear,roth1991maximum} over finite fields.
In \cite{augot2013rank}, the definition was extended to certain fields of characteristic zero, using $\autom$-polynomials instead of linearized polynomials.

\begin{definition}
Let $g_1,\dots,g_n \in \L$ be linearly independent over $\K$. Then a Gabidulin code of length $n$ and dimension $k \leq n$ is defined as
\begin{align*}
\CGab[n,k] = \left\{ (f(g_1), \dots, f(g_n)) \, : \, f \in \Lset \, \land \, \deg f < k  \right\}.
\end{align*}
\end{definition}

An overview of properties can be found in \cite{robert2015phd}.

\section{A New Decoding Approach}
\label{sec:Decoding}

In the following, let $\L, \K$ and $\autom \in \GalLK$ be such that $\charpoly_\autom$ is square-free.
We assume that $\autom(\cdot)$ can be computed in $O(1)$ operations in $\L$.
Under these assumptions, the latter only being important for complexity statements, we show that the decoding problem is similar to the finite field case.

Suppose that a codeword $\c \in \CGab$ is corrupted by an error $\e = (e_1,\dots,e_n) \in \L^n$ of rank weight $\terr := \wtR(\e)$. The \emph{received word} is then given by
\begin{align*}
\r = \c + \e  \in \L^n.
\end{align*}
We say that $\terr$ errors occurred.
The goal of decoding is to recover $\c$ from $\r$ if $\terr$ is not too large.

\subsection{Key Equation}

\begin{definition}\label{def:elp}
We define the error span polynomial
\begin{align*}
\Lambda = \MSP{\LH{e_1,\dots,e_n}}.
\end{align*}
\end{definition}

The following lemma is, in contrary to the finite field case, not obvious (cf. \cref{thm:existence_annihilator}) and only holds for the case of $\charpoly_\autom$ being square-free.

\begin{lemma}\label{lem:deg_Lambda}
$\deg \Lambda = \terr$
\end{lemma}

\begin{proof}
This follows directly from \cref{thm:existence_annihilator} together with $\deg \Lambda = \dim \LH{e_1,\dots,e_n} = \wtR(\e) = \terr$.
\end{proof}

The following lemma is necessary to prove \cref{thm:key_equation}, the main statement of this section.

\begin{lemma}\label{lem:ab_equiv_modulo_MSP}
Let $\U \subseteq \L$ be a $\K$-subspace and $a,b \in \Lset$.
\begin{align*}
a \equiv b \mod \MSP{\U} \quad \Leftrightarrow \quad a(u) = b(u) \; \;  \forall \, u \in \U
\end{align*}
\end{lemma}

\begin{proof}
By \cref{lem:division}, there are $\quo,\rem \in \Lset$ with
\begin{align*}
a-b = \quo \cdot \MSP{\U} + \rem
\end{align*}
and $\deg \rem < \deg \MSP{\U}$.
Then,
\begin{align*}
a(u) &= b(u) \;  \forall \, u \in \U \\
\Leftrightarrow \quad a(u)-b(u) &= (a-b)(u) = (\quo \cdot \MSP{\U} + \rem)(u) \\
&= \quo(\MSP{\U}(u)) + \rem(u) = \quo(0)+\rem(u) \\
&= \rem(u) = 0 \;  \forall \, u \in \U.
\end{align*}
Also, $\rem(u) = 0$ for all $u \in \U$ if and only if $\rem=0$, since otherwise it would contradict the minimality of $\MSP{\U}$.
\end{proof}

Let $\rhat$ be the known interpolation polynomial of degree $\deg \rhat < n$ corresponding to the received word $\r$ as in Theorem~\ref{thm:interpolation}.
Recall that $f$ is the unknown information polynomial of degree $\deg f < k$ and $\Lambda$ is the unkown error span polynomial.
Also, $\Gpoly$ is known and has degree $\deg \Gpoly = n$, since the $g_i$'s are linearly independent.
The following statement is an analogue to Gao's key equation for Reed--Solomon codes and a generalization of \cite[Theorem~3.6]{wachter2013decoding}, where it was proven for finite field Gabidulin codes.

\begin{theorem}[Key Equation]\label{thm:key_equation}
\begin{align}
\Lambda \cdot \rhat \equiv \Lambda \cdot f \mod \MSP{\LH{g_1,\dots,g_n}} \label{eq:key_equation}
\end{align}
\end{theorem}

\begin{proof}
Let $u \in \LH{g_1,\dots,g_n}$. Then, we can write $u$ as a $\K$-linear combination of the $g_i$'s, $u = \sum_{i=1}^{n} \alpha_i g_i$, and
\begin{align*}
&(\Lambda \cdot \rhat)(u) - (\Lambda \cdot f)(u) = \Lambda(\rhat(u)-f(u)) \\
&=\Lambda(\rhat(\smallsum_{i=1}^{n} \alpha_i g_i)-f(\smallsum_{i=1}^{n} \alpha_i g_i)) \\
&= \smallsum_{i=1}^{n} \alpha_i \Lambda(\rhat(g_i)-f(g_i)) 
= \smallsum_{i=1}^{n} \alpha_i \Lambda(r_i-c_i) \\
&= \smallsum_{i=1}^{n} \alpha_i \Lambda(e_i)
= 0.
\end{align*}
The statement follows by \cref{lem:ab_equiv_modulo_MSP}.
\end{proof}

\subsection{Decoding using Shift Register Synthesis Problems}

Since it is hard to directly find a solution to the key equation, which is non-linear, we try to find a solution to the following shift register synthesis problem, which is formulated in a similar way as the problem which is solved in \cite{fitzpatrick1995key} over ordinary polynomial rings.

\begin{definition}
Let $k$, $\rhat$ and $\Gpoly$ be given as above.
A \emph{shift register problem} (SRP) is the problem of finding $(\lambda,\omega)~\in~(\Lset^*)^2$ such that
\begin{align}
\lambda \rhat \equiv \omega &\mod \Gpoly \label{eq:congruence} \\
\deg \lambda &> \deg \omega + k \label{eq:degree}\\
\deg \lambda &\text{ minimal} \label{eq:minimality}
\end{align}
\end{definition}

The following theorem is, besides the key equation, the main statement of this paper.
It proves that the decoding problem and the SRP are equivalent if the number of errors is less than half the minimum distance.

\begin{theorem}\label{thm:key_equation_and_SRP}
If $\tau < \frac{d}{2}$, the SRP has a solution $(\lambda,\omega)$ and any such solution satisfies
\begin{align*}
(\Lambda, \Lambda f) = \alpha (\lambda,\omega)
\end{align*}
for some $\alpha \in \L^*$, minimum distance $d$ and information polynomial $f~\in~\Lset$.
\end{theorem}

\begin{proof}
We first prove that the SRP has a solution and all solutions satisfy $\omega = \lambda f$, by applying similar arguments as in the proof of \cite[Theorem~25]{robert2015phd}.
Then we show that the solution is unique up to a scalar multiplication.
By \cref{thm:key_equation}, $(\Lambda, \Lambda f)$ fulfills the congruence relation \eqref{eq:congruence} and due to
\begin{align*}
\deg \Lambda f = \deg \Lambda + \deg f < \deg \Lambda + k,
\end{align*}
it also satisfies the degree condition \eqref{eq:degree}.
Thus, the SRP has a solution\footnote{Either $(\Lambda,\Lambda f)$ or a ``smaller'' solution in terms of $\deg \lambda$} $(\lambda,\omega)$, and by \cref{lem:deg_Lambda}, any such solution satisfies
\begin{align}
\deg \lambda &\leq \deg \Lambda = \tau, \label{eq:lambda_degree} \\
\deg \omega &< \deg \lambda + k = \tau + k. \label{eq:omega_degree}
\end{align}
We also know that $\dim \LH{e_1,\dots,e_n} = \tau$, implying
\begin{align*}
\deg \MSP{\LH{\lambda(e_1),\dots,\lambda(e_n)}} = \dim \LH{\lambda(e_1),\dots,\lambda(e_n)} \leq \tau
\end{align*}
and thus,
\begin{align*}
\deg \MSP{\LH{\lambda(e_1),\dots,\lambda(e_n)}} (\omega-\lambda f) < \deg \MSP{\LH{\lambda(e_1),\dots,\lambda(e_n)}} + \tau + k \\
\leq 2 \tau + k \leq 2 \tfrac{d-1}{2} + k = n-k+k = n,
\end{align*}
Due to \eqref{eq:congruence} and Lemma~\ref{lem:ab_equiv_modulo_MSP}, $\lambda(\rhat(g_i)) = (\lambda \rhat)(g_i) = \omega(g_i)$ for all $i$, we obtain
\begin{align*}
&\MSP{\LH{\lambda(e_1),\dots,\lambda(e_n)}} (\omega-\lambda f) (g_i) \\
&= \MSP{\LH{\lambda(e_1),\dots,\lambda(e_n)}} (\omega(g_i)-\lambda(f(g_i))) \\
&= \MSP{\LH{\lambda(e_1),\dots,\lambda(e_n)}} (\lambda(\rhat(g_i))-\lambda(f(g_i))) \\
&= \MSP{\LH{\lambda(e_1),\dots,\lambda(e_n)}} (\lambda(r_i-c_i)) \\
&= \MSP{\LH{\lambda(e_1),\dots,\lambda(e_n)}} (\lambda(e_i)) = 0
\end{align*}
Thus, we obtain $\MSP{\LH{\lambda(e_1),\dots,\lambda(e_n)}} (\omega-\lambda f) = 0$ because the polynomial has degree $<n$ but evaluates to $0$ at $n$ linearly independent positions (cf. \cref{lem:dimker_deg_bound}).
Since $\Lset$ is an integral domain, we get $\omega = \lambda f$.

Together with the congruence relation \eqref{eq:congruence},
it follows that
\begin{align*}
\lambda(\rhat-f) \equiv 0 \mod \Gpoly,
\end{align*}
thus, $\lambda(e_i) = \lambda(\rhat-f)(g_i) = 0$ $\forall i=1,\dots,n$.
Due to $\deg \lambda \leq \deg \Lambda$, $\lambda$ must be the annihilator polynomial $\Lambda$ of $\LH{e_1,\dots,e_n}$ multiplied by a scalar $\alpha^{-1} = \lambda_{\deg \lambda} \in \L^*$, the leading coefficent of the polynomial $\lambda$.
Hence, also
\begin{align}
\alpha \omega = \alpha \lambda f = \Lambda f \label{eq:alpha}
\end{align}
and the claim is proven.
\end{proof}

\begin{remark}
In the case $\tau < \frac{d}{2}$, a solution of the SRP is also a solution to the \emph{linear reconstruction problem} discussed in \cite{robert2015phd}.
This follows from the degree conditions \eqref{eq:lambda_degree} and \eqref{eq:omega_degree}, and the observation that $\lambda(\rhat(g_i)) = \omega(g_i)$ for all $i=1,\dots,n$.
\end{remark}

We can conclude that for rank errors up to half the minimum distance ${\mathrm{d}_{\mathrm{R},i}(\mathbf{r},\mathbf{c}) = \omega_i(\e) < \frac{d}{2}}$, using any rank metric $\mathrm{d}_{\mathrm{R},i}$ with $i \in \{1,2,3,4\}$ of Definition~\ref{def:rank_metrics}, we can solve the decoding problem by finding a solution of the SRP since the number of errors is $\tau = \omega_1(\e) \leq \omega_i(\e) < \frac{d}{2}$ (cf. Lemma~\ref{lem:rank_weights}).
Note that certain Gabidulin codes over finite fields cannot be decoded beyond half the minimum distance in polynomial time (cf.~\cite{raviv2015some}).
Investigating whether this is also true over fields of characteristic zero is beyond the scope of this paper.
The next section summarizes known algorithms to solve SRPs.

\subsection{Solving Shift Register Problems}
\label{subsec:solvingSRP}

SRPs over $\L[x]$ and $\Lset$ are well-studied and have been used for decoding of several algebraic codes, including Reed--Solomon and (finite field) Gabidulin codes.

Two of the most important algorithms to solve these kinds of problems are:
\begin{enumerate}
\item The \emph{Extended Euclidean Algorithm}. 

Since $\Lset$ is a Euclidean domain, it admits a Euclidean algorithm.
It is shown e.g. in \cite{wachter2013decoding} that the Euclidean algorithm over $\FF[x;\cdot^q]$ can be performed in $O(\mathcal{D}(n))$ time, where $\mathcal{D}(n)$ is the complexity of dividing two polynomials in $\FF[x;\cdot^q]$.
These results directly translate to $\Lset$.

Using the classical division algorithm, $\mathcal{D}(n) \in O(n^2)$.
However, it is justifiable that the division method described in \cite{puchinger2015fast} generalizes to $\Lset$ where $\autom(\cdot)$ can be computed in $O(1)$, implying $\mathcal{D}(n) \in O(n^{1.69} \log(n))$.

\item \emph{Module Minimization}.

The algorithms described in \cite{li2015solving} solve a generalized version of the SRP described in this paper.
If $\autom(\cdot)$ can be computed in $O(1)$, the complexity of finding a solution of the SRP becomes
$O(n^2)$.
Moreover, as already mentioned in \cite{li2015solving}, there is the substantiated hope for similar speed-ups as in the $\L[x]$ case, such as the divide-and-conquer variant described in \cite{alekhnovich_linear_2005}.
\end{enumerate}

Alternatively, a variant of the Berlekamp--Massey algorithm (cf.~\cite{richter2004error}) can be used, which might have advantages in practical scenarios.

\subsection{Issues Besides Complexity}

Since we are dealing with infinite fields, we have to deal with some issues that do not appear in the finite field case.

As already mentioned in \cite{robert2015phd}, when computing in exact computation domains, such as number fields, we have to face the problem of coefficient growth.
Fortunately, our proposed decoding method reduces the decoding problem to a problem that was already studied in terms of coefficient growth before (cf. \cite{beckermann2006fraction}).
As described in Section~\ref{subsec:solvingSRP}, we can use module minimization to obtain a solution of the SRP.
More precisely, in \cite{li2015solving} a solution of the SRP is obtained by transforming a basis of a certain $\Lset$-module into a normal form, called \emph{weak Popov form}.
Instead of using the algorithms described in \cite{li2015solving} to obtain a weak Popov form, we can use the methods from \cite{beckermann2006fraction}.
The algorithms in \cite{beckermann2006fraction} are slower than those in \cite{li2015solving}, but have a better control of coefficient growth in intermediate results using fraction-free methods.

On the other hand, especially in the application of LRMR, it might be advantageous in terms of complexity not to use exact but approximate computations.
Thus, one has to deal with numerical issues.
In the Hamming metric analogy, this problem was already investigated for complex Reed--Solomon codes (cf.~\cite[Chapter~7]{zorlein2015}).
There, it turned out that a modification of the Berlekamp--Massey algorithm is the numerically most stable one among the classical approaches for solving an SRP.
It should also be noted that the interpolation algorithm presented in \cite{liu2014kotter} is a reasonable choice to compute $\rhat$, since it is the skew polynomial analogue of the numerically stable Newton interpolation with divided differences.

\subsection{Summary of the Decoding Algorithm}

Algorithm~\ref{alg:decode} summarizes the decoding procedure.

\printalgoIEEE{
\DontPrintSemicolon
\KwIn{$\r = \c + \e$}
\KwOut{$f$ such that $\c = (f(g_1),\dots,f(g_n))$\\ or ``decoding failure''.}
Calculate $\rhat$ as in Theorem \ref{thm:interpolation} \; \label{line:dec_1}
Calculate $\Gpoly$ as in Definition \ref{def:elp} \; \label{line:dec_2}
$(\lambda,\omega) \gets $ Solve SRP with input $\rhat$, $\Gpoly$ \; \label{line:dec_3}
$(\Lambda,\Omega) \gets \alpha^{-1} (\lambda,\omega)$ with $\alpha$ as in \eqref{eq:alpha} \; \label{line:dec_4}
$(\quo,\rem) \gets $ Right-divide $\Omega$ by $\Lambda$ \label{line:dec_5} (cf. Lemma~\ref{lem:division}) \;
\If{$\rem=0$}{
\Return{$\quo$}  \label{line:dec_6}}
\Else{
\Return{``decoding failure''}}
\caption{Decode Gabidulin Codes}
\label{alg:decode}
}

\begin{theorem}
Alg.~\ref{alg:decode} is correct and has complexity $O(n^2)$.
\end{theorem}

\begin{proof}
Correctness follows from Theorems~\ref{thm:key_equation} and \ref{thm:key_equation_and_SRP}.
The lines of the algorithm have the following complexities, implying the overall statement:
\begin{itemize}
\item Line~\ref{line:dec_1}: We can use the interpolation algorithm for skew polynomials presented in \cite{liu2014kotter}, having complexity $O(n^2)$.
\item Line~\ref{line:dec_2}: $O(n^2)$ by Algorithm~\ref{alg:annihilator}.
\item Line~\ref{line:dec_3}: $O(n^2)$ using e.g. module minimization as in \cite{li2015solving}.
\item Line~\ref{line:dec_4}: Negligible.
\item Line~\ref{line:dec_5}: $O(n^2)$ using the standard algorithm \cite{ore1933theory}. \qedhere
\end{itemize}
\end{proof}

\section{Conclusion}
\label{sec:Conclusion}

We have proposed a new method for decoding Gabidulin codes over fields with characteristic zero, reducing the decoding complexity to $O(n^2)$ compared to $O(n^3)$ in \cite{robert2015phd}.
This alternative procedure reduces decoding to a linear shift register synthesis problem, which can be efficiently solved using several known algorithms, each having advantages in terms of speed, coefficient growth or numerical stability.
The presented work can be used for applying Gabidulin codes over characteristic zero to space-time coding and to the low-rank matrix recovery problem.
The latter one is, to the best of our knowledge, a new application for these codes.

\bibliographystyle{IEEEtran}
\bibliography{main}

\end{document}